\documentclass[reprint,aps,pra,]{revtex4-2}

\usepackage{graphicx}
\graphicspath{ {./img/} }
\usepackage{dcolumn}
\usepackage{tikz}
\usepackage{taisanul}
\usepackage{braket}
\usepackage{bm}

\begin{document}

\preprint{APS/123-QED}

\title{Does Entanglement Correlation in Ground State Guarantee Quantum Energy Teleportation?}

\author{Taisanul Haque}
    \email{taisanul.haque@stud.uni-goettingen.de}

\affiliation{%
 Institute of Theoretical Physics, University of Goettingen
}%

\date{\today}

\begin{abstract}
Although extraction of energy from the ground state is forbidden, one can utilize Quantum Energy Teleportation (QET) protocol for energy extraction --- a two-step protocol involving quantum measurements followed by LOCC.  This is an unique method to ``extract energy from ground states'' of quantum systems. QET requires some correlation in the ground state, and entanglement correlation plays a crucial roles as a resource. The general belief is that if the ground state is quantum-correlated via entanglement for two different sites in a quantum system, and if we perform measurements on one of the sites, we can find an LOCC for the other site to successfully accomplish QET. In this paper, we show that this belief may not be true in the case of the Toric Code. We demonstrate this by performing a PVM measurements on spins in the Toric Code. Based on the measurement outcomes, we found that there is no LOCC for successful QET.
\end{abstract}

\maketitle


\section{\label{sec:level1}Introduction}
Quantum Energy Teleportation (QET) is a novel protocol enabling the ``extraction of energy from the ground state of quantum systems''\cite{Hotta_2008,Hotta_2009,Hotta_2010,hotta2011quantumenergyteleportationintroductory}, a process that is classically inconceivable due to the prohibition of energy extraction from systems in their lowest-energy state. The QET protocol, a two-step mechanism involving local quantum measurements and subsequent local operations and classical communication (LOCC), exploits quantum correlations inherent in the ground state to redistribute and effectively extract energy without violating fundamental conservation laws. Unlike conventional energy transfer mechanisms reliant on physical carriers or direct interactions, QET uses the nonlocality of quantum correlations, making it a promising candidate for energy transfer in regimes where classical approaches falter.

\vspace{0.5cm}

QET has been extensively studied in various quantum systems, extending beyond its initial theoretical formulation. Specifically, QET protocols have been analyzed in diverse physical models, demonstrating their applicability across multiple quantum platforms. For instance, QET has been investigated in the harmonic chain model \cite{PhysRevA.82.042329}. Similarly, QET has been explored in spin chain models \cite{Hotta_2009}, providing insights into energy extraction mechanisms in discrete many-body quantum systems governed by spin interactions.
\vspace{0.5cm}

Moreover, the feasibility of QET has been examined in trapped ion systems \cite{PhysRevA.80.042323}, where the controlled interactions between ions allow for precise manipulation of quantum correlations necessary for energy teleportation. In the realm of condensed matter physics, QET has been theoretically proposed in quantum Hall systems \cite{PhysRevA.84.032336}. Remarkably, extensions of QET to black hole physics have also been studied \cite{PhysRevD.81.044025}, offering intriguing implications for quantum field theory in curved spacetime and quantum gravity.

Beyond theoretical advancements, recent experimental progress has provided empirical validation of QET. Notably, QET was experimentally realized in a controlled quantum system \cite{Rodriguez-Briones:2022jla}, marking a significant milestone in the practical implementation of quantum energy transfer protocols. Furthermore, QET protocols have been demonstrated using quantum hardware \cite{PhysRevApplied.20.024051}, using quantum computing platforms to implement and test energy teleportation mechanisms in engineered quantum circuits. These developments underscore the growing experimental feasibility of QET and its potential applications in future quantum technologies.

A crucial requirement for QET is the presence of quantum correlations, such as entanglement, in the ground state. These correlations enable the protocol to redistribute energy within the system upon measurement, transferring it from one region to another. However, the precise role of entanglement and the conditions under which QET is feasible remain active areas of research, particularly in systems with complex topological structures or unconventional correlation patterns. Moeover, it has been shown\cite{Frey_2013,haque2024aspectsquantumenergyteleportation} that entanglement is not a necessary resource for QET.

In this paper, we address the validity of QET in the context of the Toric Code, a model in topological quantum computing characterized by long range entanglement and robust ground-state degeneracy\cite{Kitaev_2003,Kitaev_2006,KitaevTEE_2006}. By analyzing the outcomes of specific projective measurements and subsequent LOCC, we claimed that QET may not be feasible in spite of having entanglement in ground states. Our findings reveal fundamental constraints on the protocol’s applicability, offering new insights into the interplay between quantum correlations and energy transport in quantum systems.

\section{Quantum Energy Teleportation Protocol}

Consider a quantum system described by the Hamiltonian
$$
{H} = \sum_{i} {h}_i,
$$
where ${h}_i$ represents nearest neighbor interacting local Hamiltonian at site $i$ , and without loss of generality, the ground state $\ket{\psi_0}$ satisfies ${H} \ket{\psi_0} = 0 \ket{\psi_0}$.

\subsection{Local Measurement}
In the QET protocol, a measurements are performed on a subsystem $A$ of the full system. The need for measurement is to breaking down the passivity of ground state, thus after local measurements in A, the whole system would be in an excited state soley due to energy deposition in A. This measurement should not increase any energy density where LOCC is to be applied i.e., in B (Bob's sites). Generally, one would find some LOCC for Bob, which enables to locally extracts some energy from his system.

The measurement is described by $\{M_A(k)\}_k$, satisfying the completeness relation:
$$
\sum_k {M}_A(k)^{\dagger} {M}_A(k) = \mathbb{I}_A.
$$
When the outcome $k$ is obtained, the post-measurement state of the system becomes
$$
\ket{\psi_k} = \frac{{M}_A(k)\ket{\psi_0}}{\sqrt{p_k}},
$$
where the probability of the outcome $k$ is
$$
p_k = \bra{\psi_0} {M}_A(k)^{\dagger} {M}_A(k) \ket{\psi_0}.
$$

On an average, the final state after the measurements becomes $\rho_A=\sum_k \ket{\psi_k}\bra{\psi_k}$ and energy injected to the system is
\begin{align}
    E_A=\sum_k \braket{\psi_k|H|\psi_k}\label{eq:II.1}
\end{align}
\subsection{LOCC and Energy Extraction}
The local measurement on subsystem $A$ breakdowns the quantum correlation between A and B, thus their local energy density is no longer stored in a correlated configuration. Hence, Bob can apply some local unitary based on the measurement outcomes to extract energy.

After applying LOCC $\{ U_B(k) \}_k$ on Bob's site, on an average, the final state becomes
$\rho_B=\sum_k \ket{\psi'_k}\bra{\psi'_k}$, where $\ket{\psi'_k}=U_B(k)\ket{\psi_k}$

The energy in the whole system due to Bob's action becomes
\begin{align}
    E_B=\sum_k\braket{\psi'_k|H|\psi'_k}\label{eq:II.2}
\end{align}

From equation \eqref{eq:II.1} and \eqref{eq:II.2}, Bob can extracts some energy whenever $E_B-E_A<0$, which means Bob's system is loosing some energy and Bob can store this energy in an appropriate device. One important point in QET is that time scale for LOCC must much smaller than the time scale for energy diffusion after the measurement.

\section{The Toric Code Model}
The Toric code model was originally introduced by Alexei Kitaev\cite{Kitaev_2003}. It is a foundational model in the field of topological quantum computation. It provides an example of a quantum system that exhibits topological order \cite{Wen:1989iv,PhysRevLett.96.110405,Kitaev_2006}, making it robust against local perturbations and a promising candidate for fault-tolerant quantum computation\cite{RevModPhys.80.1083,Dennis_2002,Terhal_2015}. The model is based on a spin lattice on a torus, with degrees of freedom represented by qubits located on the edges of the lattice.

\subsection{The toric code Hamiltonian}

The Toric code model is defined on a two-dimensional square lattice of size $L\times L$ with qubits placed on the edges. The Hamiltonian of the Toric code consists of two types of operators: the \textbf{star operator/vertex} $A_v$ and the \textbf{plaquette operator} $B_p$. There are total of $2 L^2$ number of spins operators in the lattice i.e., $L^2$ numbers of $\sigma^z_i$ operators and $L^2$ numbers of $\sigma^x_i$ operators.
\begin{figure}[!ht]
    \centering
    \includegraphics[width=0.5\linewidth]{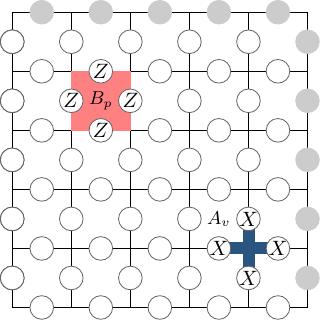}
    \caption{The $L\times L$ spin lattice shown with periodic boundary condition (PBC) in which each white circle represents spins and gray circle depicts the PBC.}
    \label{fig:1}
\end{figure}
 Plaquette operator in figure \ref{fig:1} is shown by a red block and vertex operator is shown by blue-cross. Here $X,Z$ are visually representing the Pauli's spin operators.

\subsection{Star (Vertex) Operator}
The star operator acts on the four edges surrounding a vertex (or \textit{star}) $v$. It is defined as:
$$
A_v = \prod_{i \in \text{star}(v)} \sigma_i^x
$$
where $\sigma_i^x$ is the Pauli-X matrix acting on the $i$-th qubit. The star operator enforces a constraint that the qubits around a vertex must be in an eigenstate of the product of their Pauli-X operators.

\subsection{Plaquette Operator}
The plaquette operator acts on the four edges surrounding a plaquette (or \textit{face}) $p$. It is given by:
$$
B_p = \prod_{i \in \partial p} \sigma_i^z
$$
where $\sigma_i^z$ is the Pauli-Z matrix acting on the qubit $i$, and $\partial p$ represents the boundary of the plaquette $p$. The plaquette operator enforces a similar constraint as the star operator, but using the Pauli-Z matrix.

\subsection{Full Hamiltonian}
The full Hamiltonian of the Toric code is constructed by summing over of all star and plaquette operators in the lattice:
\begin{equation}
    H = - \sum_v A_v - \sum_p B_p\label{eq:8.1}
\end{equation}
One of the important property of this Hamiltonian is that all the terms commutes i.e., all the vertex and plaquette operators commutes with each other. This is easy to see that trivially all the vertex (plaquette) operators commutes with itself, the only nontrivial part is all the vertex operators commutes with all other plaquette operators. Since, vertex and plaquette operators shares either two spins or no spins, thus if it does not shares any spins then it commutes and if it does shares two spins then two spins compensate the $-1$ and commutes.

Notice that, all the plaquette and vertex operators satisfy $B^2_p=1$ and $A^2_v=1$ due to the properties of Pauli matrices. Thus the eigenvalues of these operators are $\pm 1$. Ground state for \eqref{eq:8.1} is determined by minimising the Hamiltonian. Since, all the plaquette and vertex operators commutes, ground state is given by simultaneous eigenvectors of all the plaquette and vertex operators with eigenvalue $+1$. Let us denote ground state by $\ket{\xi}$. Now system \eqref{eq:8.1} in the ground state satisfies:
\begin{equation}
    A_v\ket{\xi}=\ket{\xi} \text{ and } B_p\ket{\xi}=\ket{\xi} \text{ for all }v, p\label{eq:8.2}
\end{equation}
In the context of quantum error correction, $A_v$ and $B_p$ for all $v,p$ are called \textit{Stabilizer operators} because it stabilizes the state $\ket{\xi}$. This state represents a long ranged entangled quantum state\cite{KitaevTEE_2006}. Ground state is degenerate and degeneracy depends on the topology of the surface on which the lattice is defined\cite{Kitaev_2003}. For a system on a torus, the degeneracy is four. This degeneracy is a global property of the system and cannot be detected by any local measurements, making the ground states robust against local perturbations.

The four degenerate ground states can be labeled by the eigenvalues of two pairs of non-local operators called the \textit{Wilson loops} $W_x$ and $W_y$, which act along the non-contractible loops around the torus in the $x$- and $y$-directions.
\begin{figure}[!ht]
    \centering
    \includegraphics[width=0.5\linewidth]{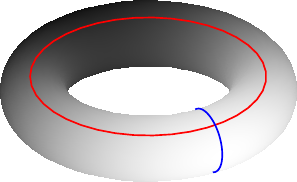}
    \caption{Two different kind of Wilson loops are illustrated. The red loop corresponds to the $W_x$ operator and blue one correspond to $W_y$.}
    \label{fig:2}
\end{figure}

These operators are defined as:
$$
W_x = \prod_{i \in \gamma_x} \sigma_i^z, \quad W_y = \prod_{i \in \gamma_y} \sigma_i^x
$$
where $\gamma_x$ and $\gamma_y$ are the loops around the torus in the $x$- and $y$-directions, respectively. The four ground states correspond to the different combinations of eigenvalues $\pm 1$ for $W_x$ and $W_y$.

\section{Excitations and Anyons}

The excitations in the Toric code model are created by violating the eigenvalue conditions of the star and plaquette operators. These excitations correspond to \textit{anyons}, which are quasiparticles that exhibit non-trivial braiding statistics\cite{Kitaev_2006}.

\subsection{Electric and Magnetic Excitations}

There are two types of anyonic excitations in the Toric code:

\begin{itemize}
    \item \textbf{Electric charges (e)}: These are created by flipping the eigenvalue of a star operator $A_v$ from $+1$ to $-1$. The electric charge corresponds to a violation of the star constraint and can be viewed as an excitation living on the vertices of the lattice.

    \item \textbf{Magnetic fluxes (m)}: These are created by flipping the eigenvalue of a plaquette operator $B_p$ from $+1$ to $-1$. The magnetic flux corresponds to a violation of the plaquette constraint and can be viewed as an excitation living on the faces of the lattice.
\end{itemize}

The electric charges and magnetic fluxes are mutual anyons, meaning that when one type of particle is moved around another, the quantum state acquires a non-trivial phase factor i.e., Berry phase\cite{PhysRevLett.48.1144,Kitaev_2003}. This behavior is solely due the topological nature of the quantum system.

\section{Quantum energy teleportation?}
Ground state of the toric code Hamiltonian \eqref{eq:8.1} has topological order parameters and has long range entanglement. In the paper\cite{PhysRevA.71.022315}, these authors have demonstrated different possible way divide the toric code system into bipartite system to quantify the entanglement by entanglement entropy. For example; A single spin is maximally entangled with rest of the system whereas two spins located at distinct lattice sites are not entangled at all. Also, a spin chain system is entangled to it's complement system.

Here, we perform the QET protocol for a bipartite system in which system A in the figure \ref{fig:3} contains $2L^2-1$ spins and B contains $1$ spin.

\begin{figure}[!ht]
    \centering
    \includegraphics[width=0.5\linewidth]{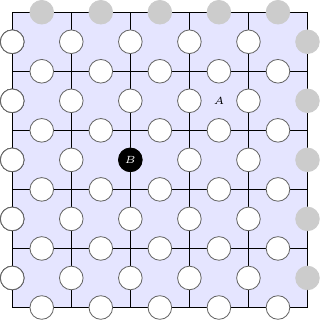}
    \caption{A bipartite system with $2L^2-1$ spins in A i.e., spins in the blue colored regions and a single spin in B with black dot.}
    \label{fig:3}
\end{figure}

In the toric lattice, we notice that each spins are shared by two plaquette operators and two vertex operators. Thus, any measurements operators commutes with vertex and plaquette operators if it does not measures spins on these vertex and plaquette operators because they are spatially located at different sites.

\subsection{Measurement}
Alice performs a projective measurements on $n$ number of spins for all $1\le n\le 2L^2-1$ in A. Kraus operator for measurements are, $M_A(k)=\frac{1}{2}(\mathbf{I}+k\sigma^x_{\pi(2)}\sigma^x_{\pi(3)}\cdots \sigma^x_{\pi(n)})$ where  $\pi:\{2,3,\cdots, 2L^2\}\ra \{2,3,\cdots, 2L^2\}$ is a bijective map i.e., $\pi$ is permutation map. Index $1$ is reserved for system B. Notice that, $M_A(k)$ satisfy all the criteria for measurement operator for $k=\pm 1$. Moreover, $M^2_A(k)=M_A(k)M^\dagger_A(k)=M_A(k)$.

Now in the ground state, $\ket{\xi}$, each Pauli-X operators creates two excitation on every plaquette operators and vertex operators are unaffected. We can also measures the same spins with respect to vertex operators by defining the measurements operators with Pauli-Z matrices. In either cases, final results will be similar due to similar properties of these operator. One can also choose Pauli-Y measurements, the final results does not change. The energy injected into the system due to the measurement is given by
\begin{equation}
    E_A=\sum_k\braket{\xi|M_A(k)HM_A(k)|\xi}\label{eq:8.3}
\end{equation}
We observe that, measurements on spins of open strings within A always creates two excitation i.e., two magnetic vortices. In our measurement scheme, in general we perform maximum of $n$ joint spins measurements in A. If we choose a permutation $\pi$ such that there are $m<n$ numbers of identity maps in $M_A(k)$ then physically we measure $n-m$ spins if $m$ is even otherwise $n-m+1$ spins. For example; if $\pi: (2)(3)(4)(567\cdots, n)$ in the cycle notation, then there are product of three Pauli-X matrices in $M_A(k)$ which is equivalent to single Pauli-X matrix. Thus, in this case we measure $n-2$ physical qubit. In our case, we measure all the spins in A to make sure it breaks all the correlation between A and B. Thus we see that, only two vortices are created on the plaquette operators in which spin B is located.

\subsection{LOCC}
Based on the measurement outcomes, Bob performs a general local unitary, $U_B=\cos(\theta)+ik\sin(\theta)\hat{n}\cdot\vec\sigma_B$, to his system where $\hat{n}=(n_x,n_y,n_z)$ is an unit vector. This means, once Bob knows the measurement outcome from Alice, he instantly perform local rotation of his spins by angle $\theta$ around the axis $\hat n$. 

Now energy after LOCC becomes:

\begin{equation}
    E_B=\sum_k\braket{\xi|M_A(k)U^\dagger_B(k)HU_B(k)M_A(k)|\xi}\label{eq:8.4}
\end{equation}

To evaluate the above expression, we calculate $[H,U_B]=ik\sin(\theta)[H,\hat{n}\cdot\vec\sigma_B]$, which implies $U^\dagger_BHU_B=H + ik\sin(\theta)U^\dagger_B[H,\hat{n}\cdot\vec\sigma_B]$, therefore
\begin{widetext}
    \begin{equation}
    M_A(k)U^\dagger_B(k)HU_B(k)M_A(k)= M_A(k)HM_A(k)+ik\sin(\theta)M_A(k)U^\dagger_B(k)[H,\hat{n}\cdot\vec\sigma_B]M_A(k)\label{eq:8.5}
\end{equation}
\end{widetext}

Hence, 
\begin{equation}
    E_B=E_A + \sum_kik\sin(\theta)\braket{\xi|M_A(k)U^\dagger_B(k)[H,\hat{n}\cdot\vec\sigma_B]M_A(k)|\xi}\label{eq:8.6}
\end{equation}

We show that the second term in equation \eqref{eq:8.6} is non -negative. Thus no energy teleportation!
\begin{figure}[!ht]
    \centering
    \includegraphics[width=0.5\linewidth]{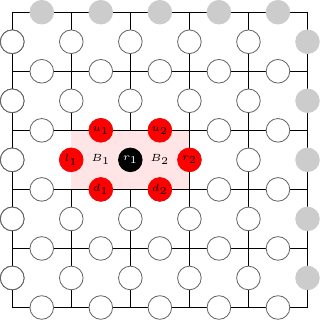}
    \caption{The red dots are the potential spins in our measurement scheme in which the measurement operator does not commutes with local unitary operators.}
    \label{fig:4}
\end{figure}

To calculate $[H,\hat{n}\cdot\vec\sigma_B]$ from the equation \eqref{eq:8.6}, let us calculate $[H,\sigma^i_B]=-\sum_v[A_v,\sigma^i_B]-\sum_p[B_p,\sigma^i_B]$ for all $i=x,y,z$. We see that $H$ commutes with all the operators except plaquette and vertex operators which contain the spin at $r_1$ in figure \ref{fig:4}. Let us denote $B=B_1+B_2$ and $A=A_1+A_2$ for the sake of calculation. Moreover, we redefined the index of $\sigma^y$ at system B as $r_1(\text{ or }l_2)$. Using, $[B,\sigma^i_{r_1}]=2B\sigma^i_{r_1}$ for $i=x,y$ and $[A,\sigma^i_{r_1}]=2A\sigma^i_{r_1}$ for $i=y,z$. Thus, the final expression for $[H,\hat{n}\cdot\vec\sigma_B]$ becomes;
\begin{equation}
    [H,\hat{n}\cdot\vec\sigma_B]=-2n_xB\sigma^x_{r_1}-2n_y(A+B)\sigma^y_{r_1}-2n_zA\sigma^z_{r_1}\label{eq:8.7}
\end{equation}

Now, we state and prove the following two lemmas, which are useful for simplification of equation\eqref{eq:8.6}:
\begin{lem}\label{lem:8.1}
    If $B$ is any plaquette operator then $M_A(k)BM_A(k)$ is either zero or $M_A(k)B$.
\end{lem}
\begin{proof}
    If $B$ is any plaquette operator such that it does not contains spins in the red circles in figure\ref{fig:4} then clearly $B$ commutes with $M_A(k)$ whenever $M_A(k)$ and $B$ contain even number of common spins. Thus using the property that $M_A(k)$ is PVM operator, we get the desired result. On the other hand, whenever $B$ and $M_A(k)$ contains odd number of common spins (red circles) then $M_ABM_A=\frac{1}{4}(\mathbf{I}+k\sigma^x_{\pi(2)}\sigma^x_{\pi(3)}\cdots \sigma^x_{\pi(n)})(\mathbf{I}-k\sigma^x_{\pi(2)}\sigma^x_{\pi(3)}\cdots \sigma^x_{\pi(n)})B=0$
\end{proof}

\begin{lem}\label{lem:8.2}
    For any Pauli matrices, $\sigma^i_j$, with $i=x,y,z$, the expectation value of $\sigma^i_j$ with respect to ground state is zero i.e., $\braket{\xi|\sigma^i_j|\xi}=0$.
\end{lem}
\begin{proof}
    Notice that when $i=x,y$, we have $\braket{\xi|\sigma^i_j|\xi}=\braket{\xi|\sigma^i_jB_j|\xi}=-\braket{\xi|B_j\sigma^i_j|\xi}=-\braket{\xi|B^\dagger_j\sigma^i_j|\xi}=-\braket{\xi|\sigma^i_j|\xi}\implies \braket{\xi|\sigma^i_j|\xi}=0$ . When $i=z$ then we use $A_j$ and rest of the calculation is similar.
\end{proof}
Expanding the equation \eqref{eq:8.6}, we can  rewrite it as 
\begin{widetext}
    \begin{align}
    E_B-E_A&=\sum_k ik\frac{\sin(2\theta)}{2}\braket{\xi|M_A(k)[H,\hat{n}\cdot \Vec{\sigma_{r_1}}]M_A(k)|\xi}+\sum_k\sin^2(\theta)\braket{\xi|M_A(k)\hat{n}\cdot \Vec{\sigma_{r_1}}[H,\hat{n}\cdot \Vec{\sigma_{r_1}}]M_A(k)|\xi}\label{eq:8.8}
\end{align}
\end{widetext}
Using the equation\eqref{eq:8.7}, the property of vertex operators, $A_jM_A(k)=M_A(k)A_j$ for all $j$, lem\ref{lem:8.1} and lem\ref{lem:8.2}, the first term in the above equation vanishes. Thus equation\eqref{eq:8.8} becomes;
\begin{equation}
    E_B-E_A=\sin^2(\theta)\sum_k\braket{\xi|M_A(k)\hat{n}\cdot \Vec{\sigma_{r_1}}[H,\hat{n}\cdot \Vec{\sigma_{r_1}}]M_A(k)|\xi}\label{eq:8.9}
\end{equation}
 From the lem \ref{lem:8.1}, we find all the expression proportional to plaquette operator $B$ in the above equation are zero. Thus only possible nonzero terms are those which are proportional to $A$. In the below, we show that all the terms of form $M_A(k)\sigma^i_{r_1}\sigma^j_{r_1}AM_A(k)$ for $(i,j)\in \{(z,x),(x,y)\}$ has zero expectation value with $\ket{\xi}$.
 \begin{lem}\label{lem:8.3}
     For all $(i,j)\in \{(z,x),(x,y)\}$, the expectation value of $M_A(k)\sigma^i_{r_1}\sigma^j_{r_1}AM_A(k)$ with $\ket{\xi}$ is zero.
 \end{lem}
\begin{proof}
    Notice that $\sigma^x_{r_1}\sigma^y_{r_1}\sigma^z_{r_1}=i$, therefore $\sigma^i_{r_1}\sigma^j_{r_1}\propto \sigma^l_{r_1}$ for some appropriate $l=y,z$. Thus $\braket{\xi|M_A(k)\sigma^i_{r_1}\sigma^j_{r_1}AM_A(k)|\xi}\propto \braket{\xi|M_A(k)\sigma^l_{r_1}AM_A(k)|\xi}$. Using $[\sigma^l_{r_1},M_A(k)]=0$, $[A,M_A(k)]=0$ and $\{\sigma^l_{r_1},A\}=0$, we find that,
    
    $\braket{\xi|M_A(k)\sigma^l_{r_1}AM_A(k)|\xi}=\braket{\xi|M_A(k)\sigma^l_{r_1}A|\xi}=-\braket{\xi|AM_A(k)\sigma^l_{r_1}|\xi}=-\braket{\xi|A^\dagger M_A(k)\sigma^l_{r_1}|\xi}=-\braket{\xi|M_A(k)\sigma^l_{r_1}A|\xi}\implies \braket{\xi|M_A(k)\sigma^l_{r_1}A|\xi}=0$
\end{proof}
Therefore, equation\eqref{eq:8.9} becomes, $E_B-E_A=4\sin^2(\theta)(n^2_y+n^2_z)\braket{\xi|\sum_kM_A(k)|\xi}=4\sin^2(\theta)(n^2_y+n^2_z)\ge 0$. Therefore it suggests, no energy teleportation!

\section{Conclusion}
In this paper we demonstrated that even if a quantum system has entangled ground state, QET may not be feasible i.e., for a given measurements scheme, which breakdowns all the quantum correlation including entanglement between bipartite system,  there may not exist any LOCC, which lowers the local energy density from global ground state. Comments are welcome.
\section{Acknowledgement}
I would like to thank Jinzhao Wang for his valuable comments on this paper.
\bibliography{apssamp}

\end{document}